\documentclass[11pt,letterpaper]{article}
\usepackage[margin=1in]{geometry}
\usepackage{amsmath,amssymb,amsfonts,setspace,url}
\usepackage{bm}
\usepackage{latexsym}
\usepackage{amsthm}
\usepackage{subfig}
\usepackage{float}
\usepackage{algorithmic}
\usepackage[ruled]{algorithm}
\usepackage{bigstrut,array,multirow}
\usepackage{here}
\usepackage{color}

\theoremstyle{plain}
\newtheorem{theorem}{Theorem}
\newtheorem{lemma}{Lemma}

\newtheorem{definition}{Definition}

\newtheorem{proposition}{Proposition}

\theoremstyle{definition}

\floatstyle{ruled}
\newfloat{algorithmfig}{thp}{lof}
\floatname{algorithmfig}{Algorithm}

\newcommand{\Omit}[1]{}

\newcounter{Codeline}

\newcommand{\Int}{\mathbb{Z}}

\newcommand{\NT}{\textsc{FindEdges}}
\newcommand{\NTT}{\textsc{FindEdgesWithPromise}}

\newcommand{\ket}[1]{|#1\rangle}

\newcommand{\E}{\mathbb{E}}
\newcommand{\Pp}{\mathcal{P}}
\newcommand{\Tt}{\mathcal{T}}

\providecommand{\Aa}{\mathcal{A}}
\providecommand{\Qq}{\mathcal{Q}}
\providecommand{\Bb}{\mathcal{B}}
\providecommand{\Cc}{\mathcal{C}}

\providecommand{\Hh}{\mathcal{H}}
\providecommand{\Nn}{\mathcal{N}}
\providecommand{\Ii}{\mathcal{I}}
\providecommand{\Jj}{\mathcal{J}}
\providecommand{\Kk}{\mathcal{K}}
\providecommand{\Vv}{\mathcal{V}}
\providecommand{\uu}{\boldsymbol{u}}
\providecommand{\vv}{\boldsymbol{v}}
\providecommand{\ww}{\boldsymbol{w}}

\begin{document}

\title{Quantum Distributed Algorithm for the All-Pairs Shortest Path Problem in the CONGEST-CLIQUE Model\\\vspace{3mm}
}
\author{Taisuke Izumi\\
Graduate School of Engineering\\
Nagoya Institute of Technology\\
\url{t-izumi@nitech.ac.jp}
\and
Fran{\c c}ois Le Gall\\
Graduate School of Informatics\\
Kyoto University\\
\url{legall@i.kyoto-u.ac.jp}
}
\date{}

\maketitle
\thispagestyle{empty}
\setcounter{page}{0}
\begin{abstract}
The All-Pairs Shortest Path problem (APSP) is one of the most central problems in distributed computation. In the CONGEST-CLIQUE model, in which $n$ nodes communicate with each other over a fully connected network by exchanging messages of $O(\log n)$ bits in synchronous rounds, the best known general algorithm for APSP uses $\tilde O(n^{1/3})$ rounds. Breaking this barrier is a fundamental challenge in distributed graph algorithms. In this paper we investigate for the first time quantum distributed algorithms in the CONGEST-CLIQUE model, where nodes can exchange messages of $O(\log n)$ \emph{quantum} bits, and show that this barrier can be broken: we construct a $\tilde O(n^{1/4})$-round quantum distributed algorithm for the APSP over directed graphs with polynomial weights in the CONGEST-CLIQUE model. This speedup in the quantum setting contrasts with the case of the standard CONGEST model, for which Elkin et al.~(PODC 2014) showed that quantum communication does not offer significant advantages over classical communication. 

Our quantum algorithm is based on a relationship discovered by Vassilevska Williams and Williams (JACM 2018) between the APSP and the detection of negative triangles in a graph. The quantum part of our algorithm exploits the framework for quantum distributed search recently developed by Le Gall and Magniez (PODC 2018). Our main technical contribution is a method showing how to implement multiple quantum searches (one for each edge in the graph) in parallel without introducing congestions. 
\end{abstract}
\newpage
\section{Introduction}
{\bf Background.}
The CONGEST-CLIQUE model is a model in distributed computing that has recently been the subject of intensive research \cite{CKKLPS15,Dolev+DISC12,Drucker+PODC14,Hegeman+PODC15,Hegeman+SIROCCO14,Hegeman+DISC14,HKN16,LP13,Lenzen+STOC11,Lotker+SPAA03,Nanongkai14,PattShamir+PODC11,Jurdzinski+18,parter18,Ghaffari+18}. In this model~$n$ nodes communicate with each other over a fully connected network (i.e., a clique) by exchanging messages of $O(\log n)$
bits in synchronous rounds. Compared with the more traditional CONGEST 
model~\cite{Peleg00}, the CONGEST-CLIQUE model removes the effect of distances between nodes in the computation and thus focuses solely on understanding the role of congestion in distributed computing. 

The study of shortest path problems is one of the central topics in the context of distributed graph algorithms. In the CONGEST model, much progress has been done in the past years \cite{LP13,HW12,Nanongkai14,HKN16,GL18,HNS17,FN18,BN18,BKKL17,Elkin17,CKKLPS15}: while for exact computation of the Single-Source Shortest Path problem (SSSP) there is still a small gap between the upper bounds and the lower bounds~\cite{SHKKANPPW12,FN18}, for the All-Pairs Shortest Path problem (APSP) an algorithm with optimal time complexity (up to possible polylogarithmic factors) has been constructed very recently~\cite{BN18}. In the CONGEST-CLIQUE model, the first non-trivial result attaining sublinear running time was an algorithm by Nanongkai~\cite{Nanongkai14}, which solves the $(2 + \epsilon)$-approximate APSP over undirected weighted graphs within $\tilde{O}(\sqrt{n})$ rounds. This was improved by Censor-Hillel at al.~\cite{CKKLPS15}, who gave a $\tilde{O}(n^{1/3})$-round exact algorithm for the general APSP (i.e., the APSP over directed graphs with polynomial weights). While faster algorithms based on fast matrix multiplication have been designed for the APSP over graphs with small weights or for approximating the shortest paths \cite{CKKLPS15,LeGall16}, 
the above $\tilde{O}(n^{1/3})$-round is still not only the best known exact algorithm for the general APSP, but also the best known exact algorithm for SSSP in the CONGEST-CLIQUE model. 
\vspace{2mm}

\noindent{\bf Quantum distributed computing.}
The power of distributed network computation in the quantum CONGEST model has been first investigated by Elkin et al.~\cite{Elkin+PODC14}. In this model the nodes can use quantum processing and communicate using quantum bits (qubits): each edge of the network corresponds to a quantum channel (e.g., an optical fiber if qubits are implemented using photons) of bandwidth $O(\log n)$ qubits. Their main conclusion was that for many fundamental problems in distributed computing, including the computation of the $s$-$t$ shortest path in weighted graphs, quantum communication does not offer significant advantages over classical communication. A significant development recently happened: Le Gall and Magniez~\cite{LeGall+PODC18} constructed a quantum distributed algorithm in the CONGEST model computing the exact diameter within $\tilde{O}(\sqrt{nD})$ rounds, where~$D$ denotes the diameter. Since Frischknecht et al.~\cite{FHW12} have shown that any classical algorithm requires~$\tilde{\Omega}(n)$ rounds, even in the case $D = O(1)$, this gives a speedup (up to quadratic when the diameter is small). At the core of this quantum algorithm lies a distributed implementation of Grover's seminal quantum algorithm \cite{GroverSTOC96}. Grover's algorithm achieves a quadratic speedup over brute-force search for generic search problems in the centralized setting. The algorithm from \cite{LeGall+PODC18} carefully adapts Grover's algorithm to the distributed CONGEST model and shows how to combine it with a classical distributed algorithm in a completely black-box way. Due to its versatility, this approach has the potential of accelerating many graph algorithms. A pressing open question is to understand for which problems in distributed computing it can actually help. \vspace{2mm}

\noindent{\bf Our result.}
While it is tempting to consider potential quantum acceleration of computing shortest paths using the distributed version of Grover's algorithm, there are several significant obstacles. The distributed quantum diameter algorithm from \cite{LeGall+PODC18} crucially relies on reducing the computation of the diameter to the search problem of finding a node with the maximum eccentricity, and this strategy does not directly work for shortest path problems. Indeed, as already mentioned, it is known that for $s$-$t$ shortest paths over weighted graphs, quantum communication cannot offer any significant speedup in the CONGEST model. Even over unweighted graphs, in the CONGEST model it is easy to extend the classical lower bound from \cite{FHW12} to show a $\tilde{\Omega}(n)$-round lower bound for the APSP that holds even in the quantum setting.

In this paper we show that a speedup is possible in the CONGEST-CLIQUE model. Our main result is the following theorem. 
\begin{theorem}\label{th:APSP}
There is a quantum algorithm that solves with high probability the All-Pairs Shortest Path problem over directed graphs with integer weights in $\{-W,\ldots,W\}$ using $\tilde O(n^{1/4}\log W)$ rounds in the CONGEST-CLIQUE model. 
\end{theorem}
As already mentioned, the best known upper bound for the APSP in the classical CONGEST-CLIQUE is due to Censor-Hillel at al.~\cite{CKKLPS15}: for graphs with integer weights in $\{-W,\ldots,W\}$ their upper bound is $\tilde{O}(n^{1/3} \log W)$ rounds. While no nontrivial lower bound is known on the classical complexity of APSP in the CONGEST-CLIQUE model, which is not surprising due to the technical challenges of proving any nontrivial lower bound in this model,
the current $\tilde{O}(n^{1/3} \log W)$ bound appears as a significant barrier for classical algorithms. Our quantum algorithm breaks this barrier. This gives strong evidence for the superiority of quantum distributed computing over classical distributed computing in the CONGEST-CLIQUE model as well (unless the barrier can be broken in the classical setting as well --- this would in itself be a significant breakthrough). Another interesting observation is that this quantum speedup occurs for a problem (the APSP) for which no quantum speedup can be achieved in the standard CONGEST model, as already mentioned.

\vspace{2mm}

\noindent{\bf Technical overview.}
The first step of our approach consists in reducing the APSP problem to the problem of detecting \emph{negative triangles} (triangles in which the sum of the weights of the three edges is negative). This reduction is inspired by the recent breakthrough by Vassilevska Williams and Williams~\cite{Williams+JACM18} in centralized algorithms that revealed the relationship between the APSP and triangle detection, via the computation of the distance product of two matrices. More precisely, our approach reduces the APSP to the problem of identifying all the edges of the graph that are involved in (at least) one negative triangle, under the promise that each edge is involved in at most $O(\log n)$ negative triangles.

In order to solve the latter problem, we would like to design an algorithm running a quadratic number of instances of negative-triangle detection simultaneously, since in the worst case $\Theta(n^2)$ edges involved in negative triangles need to be detected. Due to the fact that the query sequence generated by a single run of the distributed version of Grover's algorithm is a quantum superposition, a naive parallelization would result in high congestion of query messages, causing delays and degradation of the running time. To overcome this difficulty, we develop a novel machinery ensuring that all the parallel runs of the quantum searches are fairly load balanced, which resolves the problem of congestions. This is done by analyzing carefully the error probability of multiple quantum searches and showing that (for the problem considered) ignoring the queries that are not load balanced does not decrease significantly the success probability.

\vspace{2mm}

\noindent{\bf Other related works.}
As already mentioned, triangle detection and matrix multiplication are closely related to the APSP problem. There are several results considering those problems in the CONGEST or CONGEST-CLIQUE models~\cite{Dolev+DISC12,CKKLPS15,LeGall16,Izumi+PODC17,Pandurangan+SPAA18,Chang+SODA19,Chang+19}. In the CONGEST-CLIQUE model, in particular, an $\tilde{O}(n^{1/3})$-round algorithm for listing all triangles is proposed by Dolev et al.~\cite{Dolev+DISC12}. This algorithm is combinatorial (i.e., non-algebraic) and thus works for listing negative triangles as well. Combined with our reduction from APSP to negative triangles, this can be used to construct a classical distributed APSP algorithm in the CONGEST-CLIQUE with the same complexity  $\tilde{O}(n^{1/3} \log W)$ as the algorithm by Censor-Hillel~\cite{CKKLPS15}. While there exist faster algorithms for triangle detection~\cite{CKKLPS15,Drucker+PODC14,LeGall16}, all these faster algorithms are based on an algebraic approach (more precisely, a reduction to matrix multiplication over a ring), and cannot be used to find negative triangles (which corresponds to matrix multiplication over a semiring).

To our knowledge the present work is the first to consider the quantum CONGEST-CLIQUE model. We already mentioned the prior works \cite{Elkin+PODC14,LeGall+PODC18} on the quantum CONGEST model. Besides the vast literature on two-party quantum communication complexity (see, e.g.,~\cite{deWolf02,Broadbent+08,Denchev+08}), there exist a few works that considered other settings in quantum distributed computing. First, exact quantum protocols for leader election in anonymous networks have been developed by Tani et al.~\cite{Tani+12}. Gavoille et al.~\cite{Gavoille+DISC09} then considered quantum distributed computing in the LOCAL model, and showed that for several fundamental problems, allowing quantum communication does not lead to any significant advantage. Very recently Le Gall et al.~\cite{LeGall+STACS19} showed that there nevertheless exist some computational problems for which quantum distributed computing can be much more powerful than classical distributed computing in the LOCAL model.


\section{Preliminaries}\label{sec:prelim}
\paragraph{General notations.} Given any positive integer $p$, we use the notation $[p]$ to represent the set $\{1,2,\ldots,p\}$. Given a graph $G=(V,E)$ and any two sets $U,U'\subseteq V$, we write $\Pp(U,U')$ the set of pairs of vertices $\{u,v\}$ with $u\in U$, $v\in U'$ and $u\neq v$. When $U'=U$ we simply write $\Pp(U)=\Pp(U,U)$. Finally, for any vertex $v\in V$ we write $\Nn_G(v)$ the set of neighbors of $v$.

\vspace{-3mm}

\paragraph{Quantum CONGEST-CLIQUE model.}
Recent definitions of the quantum CONGEST model \cite{LeGall+PODC18} and the quantum LOCAL model \cite{LeGall+STACS19} are obtained by starting with the corresponding classical model (classical CONGEST model and LOCAL model, respectively) and simply allowing nodes to send quantum information instead of classical information. In this paper we use the same approach to define a natural quantum version of the CONGEST-CLIQUE model.

In the classical CONGEST-CLIQUE model, $n$ nodes communicate with each other over a fully connected network by exchanging messages of $O(\log n)$ bits in synchronous rounds. All links and nodes (corresponding to the edges and vertices of $G$, respectively) are reliable and suffer no faults. Each node has a distinct identifier. In the quantum CONGEST-CLIQUE model the only difference is that the nodes can exchange quantum information: each message exchanged consists of $O(\log n)$ quantum bits instead of $O(\log n)$ bits in the classical case. In particular, initially the nodes of the network do not share any entanglement.

This paper will describe many classical algorithms and procedures that will be used for pre-processing and  post-processing (or even used inside the main quantum part as a subprocedure). We will use many times (sometimes implicitly) the following Lemma by Dolev et al.~\cite{Dolev+DISC12}.
\begin{lemma}{\cite{Dolev+DISC12}}\label{lemma:Dolev}
In the CONGEST-CLIQUE model a set of messages in which no node is the source of more than $n$ messages and no node is the destination of more than $n$ messages can be delivered within two rounds if the source and destination of each message is known in advance to all nodes.  
\end{lemma}

\paragraph{Graph-theoretic problems in the CONGEST-CLIQUE model.}
When studying graph-theoretic problems such as the APSP problem in the classical or quantum CONGEST-CLIQUE model, the input is a graph~$G=(V,E)$ consisting of $n$ nodes, i.e., the number of nodes of the graph is the same as the number of nodes of the communication network. This means that we can assign to each node of the network a distinct label $u\in V$. The input is given as follows: each node with label~$u$ of the network receives the row of the adjacency matrix of $G$ corresponding to vertex~$u$ of~$G$. The result of the computation is defined similarly: for the APSP the node with label $u$ should output the shortest distance from $u$ to all the other nodes in $G$. We refer to \cite{CKKLPS15} for details.
\section{APSP and negative triangles}\label{sec:APSP}
In this section we show how to reduce the APSP to finding all the edges involved in a negative triangle. 
We first define the latter problem and state the main technical result of this paper (Theorem \ref{th:main}). Then we show the computation of the distance product of a matrix reduces to this problem. Finally, we recall the standard reduction from APSP to the computation of the distance product and derive Theorem \ref{th:APSP} from Theorem \ref{th:main}.

\paragraph{Finding the edges in negative triangles.}\label{sub:APSP1}
Consider an undirected weighted graph $G=(V,E,f)$ with weight function $f\colon E\to\Int$. For an edge $\{u,v\}\in E$, we use the notation $f(u,v)$ instead of $f(\{u,v\})$. 
\begin{definition}
Given three vertices $u,v,w\in V$, we say that the triple $\{u,v,w\}$ is a negative triangle in $G$ if $\{u,v\}$, $\{u,w\}$ and $\{v,w\}$ are edges and the inequality $f(u,v)+f(u,w)+f(v,w)<0$ holds. 
\end{definition}

For any pair $\{u,v\}\in\Pp(V)$, we use the notation $\Gamma_G(u,v)$ to denote the number of negative triangles involving $\{u,v\}$, i.e.,
$
\Gamma_G(u,v)
=\left|\left\{
w\in V \:|\: \{u,v,w\} \textrm{ is a negative triangle in }G
\right\}\right|.
$
We simply write $\Gamma(u,v)$ when the graph $G$ is clear from the context.

We now define the main problem considered in this paper. This problem, which we denote $\NT$,  asks to compute the list of all edges involved in a negative triangle. The formal definition of the problem is as follows.

\begin{center}
\fbox{
\begin{minipage}{15 cm} \vspace{2mm}

\noindent$\NT$\\\vspace{-3mm}

\noindent\hspace{3mm} Input: an undirected weighted graph $G=(V,E,f)$ distributed among the $n$ nodes 

\noindent\hspace{15mm} of the network
(each node $u$ gets $\Nn_G(u)$)  \vspace{2mm}

\noindent\hspace{3mm} Output: each node $u$ outputs the list of all pairs $\{u,v\}\in \Pp(V)$ such that $\Gamma(u,v)>0$
\vspace{2mm}
\end{minipage}
}
\end{center}

Let us now consider the version of this problem in which we have the promise $\Gamma(u,v)=O(\log n)$ for all pairs $\{u,v\}$. It will actually be convenient to define a more general problem where there is an additional input $S\subseteq\Pp(V)$, the promise only holds for the pairs in $S$ and we only require each node to output the edges in $S$ that are involved in a negative triangle. The definition of this version with promise, which we call $\NTT$, follows.

\begin{center}
\fbox{
\begin{minipage}{15 cm} \vspace{2mm}

\noindent$\NTT$\\\vspace{-3mm}

\noindent\hspace{3mm} Input: an undirected weighted graph $G=(V,E,f)$ and a set $S\subseteq \Pp(V)$ distributed 

\noindent\hspace{15mm}
among the $n$ nodes of the network

\noindent\hspace{15mm}
(each node $u$ gets $\Nn_G(u)$ and the list of all pairs in $S$ containing $u$)  \vspace{2mm}

\noindent\hspace{3mm} Promise: $\Gamma(u,v)\le 90\log n$ for all pairs $\{u,v\}\in S$\vspace{2mm}

\noindent\hspace{3mm} Output: each node $u$ outputs the list of all pairs $\{u,v\}\in S$ such that $\Gamma(u,v)>0$
\vspace{2mm}
\end{minipage}
}
\end{center}

It is not difficult to show a randomized reduction from solving $\NT$ to solving $O(\log n)$ instances of $\NTT$. We state this reduction in the following proposition.
\begin{proposition}\label{prop:promise-to-nopromise}
Assume there exists a $T(n)$-round algorithm that solves $\NTT$ with probability at least $1-\varepsilon$ for some $\varepsilon>0$. Then there exists a $O(T(n)\log n)$-round algorithm that solves the problem $\NT$ with probability at least $1-O((\varepsilon+1/n^3)\log n)$.
\end{proposition}
\begin{proof}
Let $\Aa$ denote the $T(n)$-round algorithm for $\NTT$. We construct an algorithm for $\NT$ as follows.
\begin{itemize}
\item[1.]
$S\gets \Pp(V)$; $M\gets \emptyset$; $i\gets 0$.
\item[2.]
While $60\cdot 2^i\log n\le n$ do:
\begin{itemize}
\item[2.1.]
Sample each edge of $G$ with probability $\sqrt{\frac{60\cdot 2^i\log n}{n}}$. Let $G'$ be the subgraph of $G$ consisting only of the sampled edges.
\item[2.2.]
Apply the algorithm $\Aa$ on input $(G',S)$. Let $S'$ be the output of the algorithm.
\item[2.3.]
$S\gets S\setminus S'$;
$M\gets M\cup S'$;
$i\gets i+1$.
\end{itemize}
\item[3.]
Apply the algorithm $\Aa$ on input $(G,S)$. Let $S''$ be the output of the algorithm.
\item[4.]
Output $M\cup S''$.
\end{itemize}
Let us call Algorithm $\Bb$ the algorithm we just described. Its round complexity is $O(T(n)\log n)$.

Let us first analyze Algorithm $\Bb$ under the assumption that Algorithm $\Aa$ never makes any error. 
We will prove below by induction the following invariant for the while loop: when testing the exit condition ``$60\cdot 2^i\log n\le n$'' at Step 2 of the while loop for some value $i$, we have $\Gamma_G(u,v)\le n/2^i$ for all $\{u,v\}\in S$, and all the pairs $\{u,v\}\in\Pp(V)$ such that $\Gamma_G(u,v)>n/2^i$ are already  contained in $M$.
This shows that at the end of the while loop we have $\Gamma_G(u,v)\le n/2^c$ for all $\{u,v\}\in S$, and all the pairs $\{u,v\}$ such that $\Gamma_G(u,v)>n/2^c$ are contained in $M$, where $c$ is the smallest integer such that $60\cdot 2^c\log n> n$. Since $n/2^c< 90\log n$, the call to Algorithm $\Aa$ at Step 3 then finds all the remaining pairs involved in negative triangles and the output at Step 4 is precisely the output of $\NT$.

The loop invariant is obviously satisfied for $i=0$. Now assume that it is satisfied when $i=k$ for some $k\ge 0$ and let us consider what is happening at Step 2.2. Consider any pair $\{u,v\}\in S$. 
From the induction hypothesis we have $\Gamma_G(u,v)<n/2^i$.
Note that $\E[\Gamma_{G'}(u,v)]=\Gamma_G(u,v)\times \frac{60\cdot 2^i\log n}{n}\le60\log n$. Chernoff's bound then implies
\[
\Pr[\Gamma_{G'}(u,v)>90\log n]\le \exp\left(-\frac{60\log n}{12}\right)<\frac{1}{n^5},
\]
which means that with probability at least $1-1/n^3$ the promise required to execute Algorithm~$\Aa$ is satisfied for all $\{u,v\}\in S$.
Let us now consider a pair $\{u,v\}\in S$ such that the inequality $\Gamma_G(u,v)>n/2^{i+1}$ holds. We have
\[
\Pr[\Gamma_{G'}(u,v)=0]=\left(1-\frac{60\cdot 2^i\log n}{n}\right)^{\Gamma_G(u,v)}
<\exp\left(-\Gamma_G(u,v)\times \frac{60\cdot 2^i\log n}{n}\right)<\frac{1}{n^{30}}
\]
and thus the pair $\{u,v\}$ is included in the output $S'$ of Algorithm~$\Aa$, and thus removed from $S$ (and added to $M$) at Step 2.3, with high probability. 
This proves that the loop invariant is satisfied for $i=k+1$ as well with probability at least $1-1/n^3-1/n^{28}$.

We have thus shown that under the assumption that Algorithm~$\Aa$ never makes any error, our algorithm solves  $\NT$ with probability at least $1-c/n^3-c/n^{28}$. Since the error probability of Algorithm~$\Aa$ is at most $\varepsilon$ and $\Aa$ is applied $c+1$ times, the union bound implies that our algorithm solves the problem $\NT$ with probability at least $1-c/n^3-c/n^{28}-(c+1)\varepsilon=1-O((\varepsilon+1/n^3)\log n)$.
\end{proof}

The main technical contribution of this paper is the following theorem, which is proved in Section \ref{sec:alg}.

\begin{theorem}\label{th:main}
There is a $\tilde O(n^{1/4})$-round quantum algorithm that solves with probability $1-O(1/n)$ the problem $\NTT$  in the CONGEST-CLIQUE model.
\end{theorem}

\paragraph{From distance products to negative triangles.}\label{sub:APSP2}
We first recall the definition of the distance product of two matrices.
\begin{definition}
Let $A$ and $B$ be two $n\times n$ matrices with entries in $\Int\cup\{-\infty,\infty\}$. The distance product of $A$ and $B$, denoted $A\star B$, is the $n\times n$ matrix $C$ such that 
$
C[i,j]=\min_{k\in[n]}\{A[i,k]+B[k,j]\}
$
for all $(i,j)\in[n]\times[n]$.
\end{definition}

Vassilevska Williams and Williams \cite{Williams+JACM18} proved a reduction from the computation of the distance product of two $n\times n$ matrices $A$ and $B$ to computing the edges involved in negative triangles in a graph. We state this reduction in the following proposition.

\begin{proposition}{\cite{Williams+JACM18}}\label{prop:distance-to-triangles}
Assume that there exists a $T(n)$-round algorithm for $\NT$. Then there exists a $O(T(n)\log M)$-round algorithm that computes the distance product of any two $n\times n$ matrices with entries in $\{-M,\ldots,M\}\cup\{-\infty,\infty\}$.
\end{proposition}
\begin{proof}[Sketch of the proof]
Let $D$ be an arbitrary symmetric $n\times n$ matrix with integer entries.
Consider the undirected tripartite graph $G=(\Ii\cup \Jj\cup \Kk,E,f)$ with $|\Ii|=|\Jj|=|\Kk|=n$ and weight function $f(i,k)=A[i,k]$ for all $(i,k)\in \Ii\times \Kk$, $f(j,k)=A[k,j]$ for all $(j,k)\in \Jj\times \Kk$ and $f(i,j)=-D[i,j]$ for all $(i,j)\in \Ii\times \Jj$. Observe that a triple $\{i,j,k\}$ with $i\in \Ii$, $j\in \Jj$ and $k\in \Kk$ is a negative triangle if and only if 
\[
A[i,k]+B[k,j]<D[i,j],
\]
which implies that the pair $\{i,j\}$ is involved in a negative triangle of $G$ if and only if 
\begin{equation}
\min_{k\in[n]}\{A[i,k]+B[k,j]\}<D[i,j].
\end{equation}
Thus by finding all the pairs $\{i,j\}$ involved in a negative triangle, we learn for which pairs $\{i,j\}$ the above inequality holds. By starting with the all-zero matrix $D$ and doing binary search (adjusting each time each entry of the matrix $D$), the distance product can thus be computed by calling $O(\log M)$ times an algorithm for $\NT$. More details can be found in \cite{Williams+JACM18}. 
\end{proof}

\paragraph{From APSP to distance products and proof of Theorem \ref{th:APSP}.}\label{sub:APSP3}
We now recall how the APSP reduces to the computation of the distance product.\footnote{Our explanations focus on computing the lengths of the shortest paths. Using standard techniques (see for instance \cite{CKKLPS15}), the approach can be adapted to return the shortest paths as well, at a cost of increasing the complexity only by a polylogarithmic factor.} This is a standard reduction: we refer to, e.g., \cite{ZwickJACM02} for a reference in the centralized setting and to~\cite{CKKLPS15} for a discussion of the reduction in the CONGEST-CLIQUE model. 

Let $G=(V,E,w)$ be a weighted directed graph on $n$ vertices with no self-loop. Assume that the graph has no negative cycle.
Let us associate $V$ with the set $[n]$. The graph can be encoded as an $n\times n$ matrix $A_G$ in which 
\[
A_G[i,j]=\left\{
\begin{tabular}{ll}
0&if $i=j$,\\
$w(i,j)$& if $i\neq j$ and $(i,j)\in E$,\\
$\infty$& if  $i\neq j$ and $(i,j)\notin E$,
\end{tabular}
\right.
\]
for each $(i,j)\in [n]\times [n]$. It is easy to check that the matrix $A_G^n$, the $n$-th power of the matrix~$A_G$ with respect to the distance product, contains the distances between all pairs of vertices of $G$. Moreover, this matrix can be computed using only $O(\log n)$ matrix products. If the weights of the graph are integers in $\{-W,\ldots,W\}$, then all the finite entries of the matrices arising during the computation of $A_G^n$ are between $-nW$ and $nW$.
We summarize this result in the following proposition.

\begin{proposition}\label{prop:APSP-to-distance}
Assume that there exists a $T(n,M)$-round algorithm that computes the distance product of any two $n\times n$ matrices with entries in $\{-M,\ldots,M\}\cup\{-\infty,\infty\}$. Then there exists a $O(T(n,nW)\log n)$-round algorithm for the APSP with integer weights in $\{-W,\ldots,W\}$.
\end{proposition}

Theorem \ref{th:APSP} then follows from the reductions described in Propositions \ref{prop:promise-to-nopromise}, \ref{prop:distance-to-triangles}, \ref{prop:APSP-to-distance} and from Theorem~\ref{th:main}. The success probability of the final quantum algorithm is $1-\tilde O((\log W)/n)$, i.e., with probability $1-\tilde O((\log W)/n)$ all the nodes of the network output the correct answer.

\section{Distributed multiple quantum searches}\label{sec:quantum}
In this section we describe our quantum technique: distributed multiple quantum searches only using typical inputs.
\subsection{Distributed quantum search}\label{sub:LGM}
Here, we explain the basic framework for quantum distributed search developed in \cite{LeGall+PODC18}. \vspace{2mm}

\noindent
{\bf Description of the result.}
Let $X$ be a finite set and $g\colon X\to \{0,1\}$ be a Boolean function over~$X$. Let $u$ be an arbitrary node of the network (e.g., an elected leader). Assume that node~$u$ can evaluate the function $g$ in~$r$ rounds: assume that there exists an $r$-round classical distributed algorithm $\Cc$ such that node $u$, when receiving as input $x\in X$, outputs $g(x)$.  Now consider the following problem: node $u$ should find one element $x\in X$ such that $g(x)=1$ (or decide that no such element exists). The trivial strategy is to compute $g(x)$ for each $x\in X$ one by one, which requires $r|X|$ rounds. Le Gall and Magniez \cite{LeGall+PODC18} showed that there exists a quantum distributed algorithm that solves this problem with high probability in $\tilde O(r\sqrt{|X|})$ rounds. While this result is described in \cite{LeGall+PODC18} for the CONGEST model, it holds for the CONGEST-CLIQUE model as well.\vspace{2mm}

\noindent 
{\bf Example.}
Let us show how the quantum distributed algorithm for the diameter from \cite{LeGall+PODC18} can be described in this setting. 
Let $X=V$ be the vertex set of the graph considered. Fix an integer $d$ and define the function $g\colon V\to \{0,1\}$ as follows: for any vertex $v\in V$ we have $g(v)=1$ if and only if the eccentricity of vertex $v$ is larger than $d$. Solving the problem described in the previous paragraph enables us to decide whether the maximum eccentricity of a vertex (i.e., the diameter of the graph) is larger than $d$, and repeating this process a logarithmic number of times for different values of $d$ (chosen via binary search) enables us to compute the diameter.\footnote{Ref.~\cite{LeGall+PODC18} then observed that in the CONGEST model the function $g$ can be evaluated by computing the eccentricity of the vertex $v$ and sending the information to the node $u$, which can be done in $O(D)$ rounds, where~$D$ denotes the diameter of the graph. Thus the diameter can be computed in $\tilde O(\sqrt{n}D)$ rounds. With a few additional improvements it is possible to obtain the better bound $\tilde O(\sqrt{nD})$, see \cite{LeGall+PODC18}.}\vspace{2mm}

\noindent 
{\bf Technical details.}
The quantum distributed algorithm for search is obtained by implementing Grover's well-known quantum search algorithm \cite{GroverSTOC96} in the distributed setting. We now explain how this works.

Let us define the two sets $A^0=\{x\in X\:|\:g(x)=0\}$ and $A^1=X\setminus A^0=\{x\in X\:|\:g(x)=1\}$ and assume that $|A^1|>0$. As usual when analyzing Grover's algorithm, we make the convenient assumption $|A^1|<|X|/2$ (otherwise finding a solution is easy).
Define the following two quantum states:
$
\ket{\psi^0}=\frac{1}{\sqrt{|A^0|}}\sum_{x\in A^0}\ket{x}
$
and
$
\ket{\psi^1}=\frac{1}{\sqrt{|A^1|}}\sum_{x\in A^1}\ket{x},
$
which are the uniform superpositions over all the elements in $A^0$ and $A^1$, respectively. Let $\Hh$ denote the subspace generated by these two quantum states. Grover's algorithm starts from the quantum state $\ket{\Phi_0}=\frac{1}{\sqrt{|X|}}\sum_{x\in X}\ket{x}$ corresponding to the uniform superposition over all the elements in $X$.
Note that $\ket{\Phi_0}$ belongs to $\Hh$ and is easy to create.
Grover's algorithm then successively applies the unitary operator corresponding to $\Cc$ and then a unitary operator $U$ independent of the function~$g$.\footnote{Here the term ``unitary operator corresponding to $\Cc$" means the unitary operator corresponding to the quantum circuit obtained by converting the classical algorithm~$\Cc$ into a quantum circuit. We typically denote this unitary operator by the same symbol $\Cc$, since there is no risk of confusion. The key observation of~\cite{LeGall+PODC18} is that this conversion preserves the complexity: if $\Cc$ is an $r$-round classical algorithm then the corresponding unitary operator can be implemented in~$O(r)$ rounds.} 

A crucial property is that any state in $\Hh$ is mapped by $U\Cc$ to a state in $\Hh$, which means that for any $k\ge 0$  the state of the system after the $k$-th iteration can be written as  
$
\ket{\Phi_k}=
(U\Cc)^k\ket{\Phi_0}=\alpha_k\ket{\psi^0}+\beta_k\ket{\psi^1}
$
for some complex numbers $\alpha_k$ and $\beta_k$ such that $|\alpha_k|^2+|\beta_k|^2=1$. The analysis of Grover's algorithm shows that by choosing $k$ such that $k=O(\sqrt{|X|})$ we can guarantee that $|\beta_k|^2\approx 1$, which means that measuring the state $\ket{\Phi_k}$ gives an element $x\in A^1$ with high probability. 
 The total round complexity is thus $\tilde O(r\sqrt{|X|})$. \vspace{2mm}

\noindent 
{\bf Multiple searches.}
We now describe an easy generalization to multiple searches of the framework presented above.
Let $X$ be a finite set and $g_1,\ldots,g_m\colon X\to \{0,1\}$ be $m$ Boolean functions over~$X$, for some integer $m\ge 1$. For each $i\in[m]$ define the set $A^1_i=\{x_i\in X\:|\:g_i(x_i)=1\}$ and assume for convenience that $|A_i^1|>0$.\footnote{This can easily be enforced by adding dummy solutions. Even if a dummy solution is introduced, if there is a real solution then Grover's algorithm will output it after a few repetitions (since when there are more than one solution Grover's algorithm outputs a random solution).}  Let $u$ be an arbitrary node of the network. Assume now that node~$u$ can evaluate the functions $g_1,\ldots,g_m$ simultaneously in~$r$ rounds. More precisely, assume that there exists a $r$-round classical distributed algorithm $\Cc_m$ such that node $u$, when receiving as input $(x_1,\ldots,x_m)\in X^m$, outputs $(g_1(x_1),\ldots,g_m(x_m))$. Now consider the following problem: node $u$ should find one element in $A_1^1\times\cdots\times A_m^1$. The framework described above can easily be generalized to this problem by having node $u$ implement in parallel $m$ independent distributed quantum searches. This gives a quantum algorithm solving this problem with high probability in $\tilde O(r\sqrt{|X|})$ rounds.\vspace{2mm}

\subsection{Multiple searches only using typical inputs}

We now show a stronger result for the multiple searches problem introduced in the previous subsection: we construct a quantum algorithm that solves the problem even if the evaluation procedure is correct only on inputs close to typical inputs. The motivation for this assumption on the evaluation procedure is as follows. In a typical application (e.g., the example in Section \ref{sub:LGM}), the search domain~$X$ represents nodes of the network and evaluation should be delegated to these nodes. In this case, at each evaluation step, the node $u$ then needs to send a query to the nodes corresponding to each coordinate of the input $\bm{x}=(x_1,\ldots,x_m) \in X^m$. Then if $\bm{x} \in X^m$ is mostly dominated by a single $x \in X$, e.g., $\bm{x}=(x,x,\ldots,x,x)$, the communication link $(u, x)$ suffers high congestion due to the queries injected by $u$. In this subsection we show that in some cases such ``non-typical'' inputs~$\bm{x}$ can be completely ignored, which solves this congestion problem.

Let us first introduce the following notation: for any real number $\beta\ge 0$, let $\Upsilon_\beta(m,X)\subseteq X^m$ denote the set of all $\bm{x} = (x_1,\ldots,x_m)\in X^m$ such that for each $x \in X$ its frequency in $\bm{x}$ is at most~$\beta$ (i.e., $x$ appears at most~$\beta$ times in $\bm{x}$). 
Note that when $\beta>(1+\delta)m/|X|$ for some large enough $\delta>0$, the set $\Upsilon_\beta(m,X)$ includes the ``typical" elements of $X^m$, i.e., the elements in which the frequencies of all $x\in X$ are close to the frequencies in an element of $X^m$ chosen uniformly at random.

Suppose that instead of assuming the existence of an $r$-round algorithm $\Cc_m$ that simultaneously evaluates the functions $g_1,\ldots,g_m$ on $X^m$, we only assume that we have an $r$-round classical distributed algorithm $\tilde \Cc_m$ in which node~$u$ outputs $(g_1(x_1),\ldots,g_m(x_m))$ on an input $(x_1,\ldots,x_m)\in\Upsilon_\beta(m,X)$ but may output an error message (or an arbitrary output) on an input $(x_1,\ldots,x_m)\in X^m\setminus \Upsilon_\beta(m,X)$. 
Our key result is the following theorem, which shows that the same complexity as in Section \ref{sub:LGM} can be achieved in case $\beta$ is large enough so that $\tilde \Cc_m$ works correctly both on typical elements from $X^m$ and on the solutions of the search problem.

\begin{theorem}\label{th:search}
Assume that $|X|<m/(36\log m)$.
Assume the existence of an evaluation algorithm~$\tilde \Cc_m$, as just described, for some real number $\beta$ such that $\beta>8m/|X|$.
Finally, assume that
\[
A_1^1\times\cdots\times A_m^1\subseteq \Upsilon_{\beta/2}(m,X).
\]
 There exists a $\tilde O(r\sqrt{|X|})$-round quantum algorithm that outputs an element of $A_1^1\times\cdots\times A_m^1$ with probability at least $1-2/m^2$.
\end{theorem}
The proof of Theorem \ref{th:search} can be found in the appendix.
The basic idea behind the proof is fairly easy to describe. Let $\Qq$ denote the $\tilde O(r\sqrt{|X|})$-round quantum algorithm described in Section~\ref{sub:LGM}, which uses Algorithm~$\Cc_m$ (or more precisely, the quantum operator corresponding to $\Cc_m$). The initial state of this algorithm is the uniform superposition over all elements of $X^m$; since $\beta$ is large enough most of these elements are in $\Upsilon_\beta(m,X)$. The final state is close to the uniform superposition over all elements of $A_1^1\times\cdots\times A_m^1$; all these elements are also in $\Upsilon_\beta(m,X)$ from the assumption. The algorithm of Theorem \ref{th:search} is exactly the same as $\Qq$ but uses (the quantum operator corresponding to) $\tilde \Cc_m$ instead of (the quantum operator corresponding to) $\Cc_m$. Using $\tilde \Cc_m$ instead of~$\Cc_m$ obviously only has a negligible impact at the beginning of computation and at the end of the computation. The main technical difficulty is to show that this has no significant impact at each step of the computation as well. We show this by proving that at any step of the computation the quantum state of the system is close to its projection on the vector space spanned by the basis vectors that are in $\Upsilon_\beta(m,X)$.

\section{Detecting Negative Triangles}\label{sec:alg}
In this section we present a $\tilde O(n^{1/4})$-round quantum distributed algorithm that solves the problem $\NTT$, which proves Theorem \ref{th:main}. Through the section $G=(V,E,f)$ represents the input of $\NTT$, i.e., an undirected weighted graph on $n$ vertices that satisfies the promise $\Gamma(u,v)\le 90 \log n$ for all $\{u,v\}\in S$.

\subsection{Overall description of the algorithm}\label{sub:overall}

The description of our algorithm will use two partitions of the vertex set $V$, which we now introduce. For the ease of presentation we assume that the three numbers $n^{1/4}$, $\sqrt{n}$ and $n^{3/4}$ are integers (otherwise we can simply round them to the next integers and slightly adjust the sizes of the sets). The first partition is an arbitrary partition of $V$ into $n^{1/4}$ subsets each containing $n^{3/4}$ elements. We denote~$\Vv$ the collection of subsets making this partition. The second partition is an arbitrary partition of $V$ into $\sqrt{n}$ subsets each containing $\sqrt{n}$ elements. We denote $\Vv'$ the collection of subsets making this partition. In addition to the labeling scheme described in Section \ref{sec:prelim}, which labels the nodes of the network by elements of~$V$, our algorithm will also use the following two labeling schemes.

\paragraph{Second labeling scheme.}
Let us write $\Tt=\Vv\times \Vv\times \Vv'$ and observe that $|\Tt|=n$. We assign one distinct label from~$\Tt$ to each node of the network. We will simply write ``node $(\uu,\vv,\ww)$'' to refer to the node with label $(\uu,\vv,\ww)$, for each triple $(\uu,\vv,\ww)\in\Tt$. This labeling scheme will be used by the algorithm to decide which node should gather the information of the graph: node $(\uu,\vv,\ww)$ will gather the weights of all the edges $\{u,w\}\in \Pp(\uu, \ww)$ and $\{w,v\}\in \Pp(\ww,\vv)$ of the graph.

\paragraph{Third labeling scheme.}
For each $(\uu,\vv)\in \Vv\times \Vv$ we have $|\Pp(\uu,\vv)|=\Theta(n^{3/2})$. We will describe below a procedure that partitions the set $\Pp(\uu,\vv)$ into $\sqrt{n}$ sets each of size $\tilde\Theta(n)$. These sets will be denoted $\Lambda_x(\uu,\vv)$, for $x\in[\sqrt{n}]$.
Our third labeling scheme assigns one distinct label $(\uu,\vv,x)\in \Vv\times \Vv\times [\sqrt{n}]$ to each node of the network. Again, we will simply write ``node $(\uu,\vv,x)$'' to refer to the node with label $(\uu,\vv,x)$, for each triple $(\uu,\vv,x)\in\Vv\times \Vv\times [\sqrt{n}]$. This labeling scheme will be used by the algorithm to distribute to search for triangles: node $(\uu,\vv,x)$ will be in charge of checking the existence of all the triangles involving one edge in the set $\Lambda_x(\uu,\vv)$.

\paragraph{The partition procedure.}
We now describe how to construct the sets $\Lambda_x(\uu,\vv)$. For technical reasons it will be much more convenient to use a covering instead of a partition of $\Pp(\uu,\vv)$, i.e., to allow some elements to appear more than once, and to construct the covering randomly rather than deterministically.

Consider the following process. Each node $(\uu,\vv,x)\in \Vv\times \Vv\times [\sqrt{n}]$ constructs the set $\Lambda_x(\uu,\vv)\subseteq \Pp(\uu,\vv)$ as follows: starting with the empty set, each pair $\{u,v\}\in \Pp(\uu,\vv)$ is added by the node to its set $\Lambda_x(\uu,\vv)$ with probability $10\log n/\sqrt{n}$. We say that the set $\Lambda_x(\uu,\vv)$ is $\emph{well-balanced}$ if the inequality
$
\big|\{v\in \vv \:|\: \{u,v\}\in\Lambda_x(\uu,\vv)\}\big|\le 100\cdot n^{1/4}\log n
$ 
holds for all $u\in \uu$. The following lemma, which is proved by standard probabilistic arguments, shows that with high probability the sets created by this process are well-balanced and cover all the set $\Pp(\uu,\vv)$.

\begin{lemma}\label{lemma:cond1}
With probability at least $1-2/n$ the following statements hold 
for all $(\uu,\vv)\in \Vv\times \Vv$:
\begin{itemize}
\item[(i)]
$\Lambda_x(\uu,\vv)$ is well-balanced for each $x\in[\sqrt{n}]$;
\item[(ii)]
$\bigcup_{x\in[\sqrt{n}]}\Lambda_x(\uu,\vv)=\Pp(\uu,\vv)$.
\end{itemize}
\end{lemma}
\begin{proof}
Let us fix $(\uu,\vv)\in \Vv\times\Vv$. 

For any $x\in[\sqrt{n}]$ we have
\[
\E\left[\big|\{v\in \vv \:|\: \{u,v\}\in\Lambda_x(\uu,\vv)\}\big|\right]=10n^{1/4}\log n
\]
for each $u\in \uu$.
Chernoff's bound and the union bound imply that Condition (i) of the lemma thus holds with probability at least $1-1/n^2$.

Let $\{u,v\}$ be an arbitrary pair in $\Pp(\uu,\vv)$. For any $x\in[\sqrt{n}]$, this pair is included in $\Lambda_x(\uu,\vv)$ with probability $10\log n/\sqrt{n}$. The probability that this pair is not included in any $\Lambda_x(\uu,\vv)$ is thus
\[
\left(1-\frac{10 \log n}{\sqrt{n}}\right)^{\sqrt{n}}\le 1/n^{4}.
\]
Condition (ii) of the lemma thus holds with probability at least $1-1/n^2$.

The statement of the lemma then follows from the above analyses and the union bound. 
\end{proof}

\paragraph{Description and analysis of the algorithm.}
Our algorithm is called \textsc{ComputePairs} and described in Figure \ref{fig:algorithm3}. Let us analyze it step by step.
Step 1 requires $O(n^{1/4})$ rounds, since $|\Pp(\uu, \ww)|=|\Pp(\ww,\vv)|=O(n^{5/4})$ hold.
Step 2 performs the sampling described in Section \ref{sub:overall}, checks which sampled pairs are in $S$ and loads their weight. Step 2 can be implemented in $O(\log n)$ rounds, since communication occurs only when all the sets $\Lambda_x(\uu,\vv)$ are well-balanced.
Lemma \ref{lemma:cond1} implies that with probability at least $1-2/n$ the following two statements hold:
\begin{itemize}
\item[(a)]
Algorithm \textsc{ComputePairs} does not abort at Step 2;
\item[(b)] 
at the end of Step 2, each pair $\{u,v\}\in S$ appears at least once at some node.
\end{itemize}
Note that when the algorithm does not abort, at the end of Step 2 each node $k=(\uu,\vv,x)$ keeps at most $100 n\log n$ pairs (since the sets $\Lambda_x(\uu,\vv)$ are all well-balanced). While the exact number of remaining pairs may naturally depend on the node, in order to simplify the notation we will assume that each node keeps precisely $m=100 n\log n$ pairs.

\begin{figure}[ht!]
\begin{center}
\fbox{
\begin{minipage}{15 cm} 
\begin{itemize}
\item[1.]
Each node $(\uu,\vv,\ww)\in\Tt$ loads all the weights $f(u,w)$ and $f(w,v)$ for all $\{u,w\}\in \Pp(\uu, \ww)$ and all $\{w,v\}\in \Pp(\ww,\vv)$.
\item[2.]
Each node $k=(\uu,\vv,x)$ constructs a random set $\Lambda_x(\uu,\vv)\subseteq \Pp(\uu,\vv)$ as described in Section \ref{sub:overall}. If the set $\Lambda_x(\uu,\vv)$ is not well-balanced, then the protocol is aborted.
Otherwise, the node $(\uu,\vv,x)$ loads the weight $f(u,v)$ of all the pairs $\{u,v\}\in \Lambda_x(\uu,\vv)$ and also checks which of these pairs are in $S$. The node $(\uu,\vv,x)$ keeps only the pairs in $S$. Let $\{u^k_1,v^k_1\},\ldots,\{u^k_{m},v^k_{m}\}$ denote these pairs.
\item[3.] 
Each node $k=(\uu,\vv,x)$ executes a search to check, for each $\ell\in [m]$, if there exists some $\ww\in \Vv'$ such that there exists $w\in \ww$ for which $(u^k_\ell,v^k_\ell,w)$ is a negative triangle. The node $k$ outputs all the pairs $\{u^k_\ell,v^k_\ell\}$ for which the search is successful. 
\end{itemize}
\end{minipage}
}
\end{center}\vspace{-6mm}
\caption{Algorithm \textsc{ComputePairs}, which computes all $\{u,v\}\in S$ involved in a negative triangle.}\label{fig:algorithm3}
\end{figure}

Step 3 of Algorithm \textsc{ComputePairs} can easily be implemented in $O(\sqrt{n})$ rounds in the classical setting.
In the next subsections we prove the following statement, which shows that a quadratic speedup can be achieved in the quantum setting. 
\begin{proposition}\label{prop:quantum}
Step 3 of Algorithm \textsc{ComputePairs} can be implemented by a $\tilde O(n^{1/4})$-round quantum algorithm that succeeds with probability at least $1-O(1/n)$.
\end{proposition}
Proposition \ref{prop:quantum} combined with the analysis done in this subsection shows that Algorithm \textsc{ComputePairs} solves the problem $\NTT$ with probability at least $1-O(1/n)$ (from the union bound). Its overall complexity is $\tilde O(n^{1/4})$. This proves Theorem \ref{th:main}.

\paragraph{Overview of the proof of Proposition \ref{prop:quantum}.}
Proposition \ref{prop:quantum} is proved by applying the methodology of Section \ref{sec:quantum} to perform simultaneous quantum searches over the search space $\Vv'$. A crucial point of the analysis is to show how to implement the checking procedure in $\tilde O(1)$ rounds. Let us discuss below the main difficulties that need to be overcome.

Consider the problem of checking, for some pair $(u^k_\ell,v^k_\ell)\in \uu\times \vv$ and some fixed $\ww\in \Vv'$, whether there exists $w\in \ww$ for which $(u^k_\ell,v^k_\ell,w)$ is a negative triangle. This can be done easily as follows: node $(\uu,\vv,x)$ first sends the pair $(u^k_\ell,v^k_\ell)$ and the weight $f(u^k_\ell,v^k_\ell)$ to node $(\uu,\vv,\ww)$. Node $(\uu,\vv,\ww)$ then checks whether the inequality
\[
\min_{w\in \ww}\{f(u^k_\ell,w)+f(w,v^k_\ell)\} \le f(u^k_\ell,v^k_\ell)
\]
holds, which can be done locally from the information gathered at Step 1 of Algorithm \textsc{ComputePairs}, and sends back this information to node $(\uu,\vv,x)$. 

For each $(\uu,\vv,x)\in\Vv\times\Vv\times[\sqrt{n}]$, node $(\uu,\vv,x)$ will execute simultaneously $m$ executions of this checking procedure (one for each value of $\ell$). Each node $(\uu,\vv,\ww)$ can thus receive, in the worst case, $m\sqrt{n}=\tilde\Theta(n^{3/2})$ pairs during one call of the checking procedure, which would require $\tilde\Theta(\sqrt{n})$ rounds.   To reduce the checking cost to $\tilde O(1)$ rounds, as needed, we will partition the set~$\Tt$ into classes and use this partition to balance the load of the checking queries in order to avoid congestions. 

The partitioning of $\Tt$ is described in Section \ref{sub:division}. It will in particular identify the triples of~$\Tt$ containing many edges from $S$ involved in negative triangles. These triples are the main source for the possible congestions in the checking procedure. A simple, but crucial, observation is that there cannot exist many such triples, since the promise of $\NTT$ guarantees that the total number of negative triangles in the graph is low. This observation is the key idea on which the implementation of the load balancing is based.  
 
\subsection{Implementation of Step 3: Dividing the set $\Tt$ into classes}\label{sub:division}
Let us first introduce a crucial definition.
\begin{definition}
For any $(\uu,\vv,\ww)\in \Tt$, let $\Delta(\uu,\vv;\ww)$ be the following quantity:
\[
\Delta(\uu,\vv;\ww)
=
\big\{
\{u,v\}\in \Pp(\uu, \vv) \cap S \:|\:\exists w\in\ww\textrm{ such that $\{u,v,w\}$ is a negative triangle in $G$}
\big\}.
\]
\end{definition}

The goal of this subsection is to divide the set of triples $\Tt$ into classes according to the value of $|\Delta(\uu,\vv;\ww)|$. Since we do not know how to compute exactly this value efficiently, we actually need to define the classification based on an approximation of $|\Delta(\uu,\vv;\ww)|$ that can be computed efficiently. In Figure~\ref{fig:algorithm1} we describe a classical algorithm called \textsc{IdentifyClass} that either aborts or assign a nonnegative integer $c_{\uu\vv\ww}$ to each node $(\uu,\vv,\ww)\in \Tt$. Note that the complexity of the algorithm is $O(\log n)$ rounds: Step 1 can obviously be implemented in $20\log n$ rounds and Step 3 does not require any communication. In the case where the algorithm does not abort, we write~$\Tt_\alpha$ the set of all triples $(\uu,\vv,\ww)\in \Tt$ such that $c_{\uu\vv\ww}=\alpha$, for each integer $\alpha\ge 0$. This defines a partition of the set $\Tt$. We now show that with high probability the algorithm does not abort and the partition indeed classifies the triples according to the value of $|\Delta(\uu,\vv;\ww)|$.

\begin{figure}[ht!]
\begin{center}
\fbox{
\begin{minipage}{15 cm} 
\begin{itemize}
\item[1.]
Each node $u\in V$ selects each vertex in $\{v\in V\:|\:\{u,v\}\in S\}$ with probability $(10\log n)/n$. Let $\Lambda(u)$ denote the set of selected vertices. If $|\Lambda(u)|>20\log n$ for at least one node $u$, then the algorithm aborts. Otherwise, each node $u$ broadcasts the set $\Lambda(u)$ to all the nodes in $V$.

Let us write 
$
R=\cup_{u\in V}\{\{u,v\}\:|\:v\in \Lambda(u)\}.
$

\item[2.]
Each node $(\uu,\vv,\ww)\in \Tt$ locally computes the value 
\[
d_{\uu\vv\ww}=
\Big|\big\{
\{u,v\}\in \Pp(\uu, \vv) \cap R \:\:|\:\:\exists w\in\ww\textrm{ s.t. $\{u,v,w\}$ is a negative triangle in $G$}
\big\}\Big|
\]
and computes the integer 
$c_{\uu\vv\ww}$ defined as the smallest $c\ge 0$ such that
$
d_{\uu\vv\ww}<10\cdot 2^c \log n.
$
\end{itemize}
\end{minipage}
}
\end{center}\vspace{-6mm}
\caption{Algorithm \textsc{IdentifyClass}.}\label{fig:algorithm1}
\end{figure}

\begin{proposition}\label{prop:identify-class}
With probability at least $1-2/n$, Algorithm \textsc{IdentifyClass} does not abort and the partition $\{\Tt_\alpha\}_{\alpha\ge 0}$ satisfies the following conditions:
\begin{itemize}
\item[(i)]
for any $(\uu,\vv,\ww)\in \Tt_0$, the inequality $|\Delta(\uu,\vv;\ww)|\le 2 n$ holds;
\item [(ii)]
for any $\alpha>0$ and any $(\uu,\vv,\ww)\in \Tt_\alpha$, the inequalities $2^{\alpha-3}n \le |\Delta(\uu,\vv;\ww)|\le 2^{\alpha+1} n$ hold.
\end{itemize}
\end{proposition}
\begin{proof}[Proof of Proposition \ref{prop:identify-class}]
Let us first compute the probability that the protocol does not abort. For each node $u\in V$, let $X_u$ be the random variable representing the number of neighbors chosen by~$u$, i.e., $X_u=|\Lambda(u)|$. Observe that 
\[
\E[X_u]=\frac{10\log n \times |\{v\in V\:|\:\{u,v\}\in S\}|}{n}\le 10\log n.
\] 
Chernoff's bound implies the inequality
\[
\Pr[X_i\ge 60\log n]
<\frac{1}{n^2}.
\]
The probability that the protocol does not abort is thus $1-1/n$, from the union bound.

We now consider the probability that Conditions (i) and (ii) hold.
Let us consider an arbitrary triple $(\uu,\vv,\ww)\in \Tt$. The expectation of the random variable $\delta_{\uu,\vv,\ww}$ is 
\[
\E[\delta_{\uu,\vv,\ww}]=\frac{10\log n \times |\Delta(\uu,\vv;\ww)|}{n}.
\]
We divide our analysis into three cases.

\begin{itemize}
\item
The case where $|\Delta(\uu,\vv;\ww)|\le n/6$. Chernoff's bound shows that
\[
\Pr[\delta_{\uu,\vv,\ww}\ge 10\log n]<2^{-10\log n}<\frac{1}{n^2}.
\]
Thus $c_{\uu\vv\ww}=0$ with probability at least $1-1/n^2$.
\item
The case where $|\Delta(\uu,\vv;\ww)|> n/6$ and $|\Delta(\uu,\vv;\ww)|< 2^{c-3}n$, for some $c\ge 1$. 
Chernoff's bound implies that 
\[
\Pr[\delta_{\uu,\vv,\ww}\ge 10\cdot 2^{c-1} \log n]\le \Pr\left[\delta_{\uu,\vv,\ww}\ge 4\E[\delta_{\uu,\vv,\ww}]\right]\le \exp\left(-\frac{90\log n}{12}\right)<\frac{1}{n^2}.
\]
Thus $c_{\uu\vv\ww}\ge c$ with probability at most $1/n^2$.

\item
Finally, the case $|\Delta(\uu,\vv;\ww)|> 2^{c+1}n$ for some $c\ge 0$. 
Chernoff's bound implies that 
\begin{align*}
\Pr[\delta_{\uu,\vv,\ww}< 10\cdot 2^c \log n]\le& \Pr\left[\delta_{\uu,\vv,\ww}< \frac{1}{2}\E[\delta_{\uu,\vv,\ww}]\right] \\
\le& \exp\left(-\frac{1}{8}\times 10\log n \times 2^{c+1}\right)<\frac{1}{n^2}.
\end{align*}
Thus $c_{\uu\vv\ww}\le c$ with probability at most $1/n^2$.
\end{itemize}
We conclude that the probability that the outputs of all the nodes $(\uu,\vv,\ww)\in \Tt$ satisfy Conditions~(i) and~(ii) is at least $1-1/n$, from the union bound.

Finally, the union bound again guarantees that the probability that the protocol does not abort and all the nodes $(\uu,\vv,\ww)\in \Tt$ satisfy Conditions (i) and (ii) is at least $1-2/n$.
\end{proof}

\subsection{Implementation of Step 3: Details and proof of Proposition \ref{prop:quantum}}\label{sub:quantum}
In this subsection we describe the details of the implementation of Step 3 in Algorithm \textsc{ComputePairs}, which is the only part of the algorithm that uses quantum computation. 

The nodes first apply Algorithm \textsc{IdentifyClass}. Proposition \ref{prop:identify-class} guarantees that with probability at least $1-2/n$ Algorithm \textsc{IdentifyClass} does not abort and the partition $\{\Tt_\alpha\}_{\alpha\ge 0}$ satisfies the two conditions of the proposition. In all this subsection we will assume that this happens. 

Let us write
\[
\Tt_\alpha[\uu,\vv]=\{\ww\in \Vv'\:|\:(\uu,\vv,\ww)\in\Tt_\alpha\}
\]
for any $(\uu,\vv)\in\Vv\times\Vv$ and any $\alpha\ge 0$. We will later use the following two lemmas that are direct consequences of the bounds given in Proposition \ref{prop:identify-class}. 

\begin{lemma}\label{lemma:cond2}
With probability at least $1-1/n^2$
the inequality 
\[
|\Lambda_x(\uu,\vv)\cap \Delta(\uu,\vv;\ww)|\le 100\cdot 2^\alpha\sqrt{n}\log n
\]
holds
for all $(\uu,\vv,x)\in \Vv\times \Vv\times [\sqrt{n}]$, all $\alpha\ge 0$ and all $\ww\in\Tt_\alpha[\uu,\vv]$.
\end{lemma}
\begin{proof}
Let us fix $(\uu,\vv)\in \Vv\times\Vv$. Consider any $x\in[\sqrt{n}]$, any $\alpha\ge 0$ and any $\ww\in\Tt_\alpha[\uu,\vv]$. 
Observe that 
\[
\E[|\Lambda_x(\uu,\vv)\cap \Delta(\uu,\vv;\ww)|]=|\Delta(\uu,\vv;\ww)|\times \frac{10\log n}{\sqrt{n}}\le 10\cdot 2^{\alpha+1}\sqrt{n}\log n.
\]
Chernoff's bound implies that the inequality $|\Lambda_x(\uu,\vv)\cap \Delta(\uu,\vv;\ww)|\le 100\cdot 2^\alpha\sqrt{n}\log n$ holds with probability at most $1/n^5$. The statement of the lemma then follows from the union bound. 
\end{proof}
\begin{lemma}\label{lemma:ub}
The following inequality holds for all 
$\alpha\ge0$ and all $(\uu,\vv)\in \Vv\times \Vv$:
\[
|\Tt_\alpha[\uu,\vv]|\le \frac{720\sqrt{n}\log n}{2^{\alpha}}. 
\]
\end{lemma}
\begin{proof}
This is obviously true for $\alpha=0$. Let us now consider any $\alpha> 0$ and any $(\uu,\vv)\in \Vv\times \Vv$.
Remember that we are assuming that $\Gamma(u,v)\le 90\log n$ for all pairs $\{u,v\}\in S$. We thus have
\[
\sum_{\ww\in\Tt_\alpha[\uu,\vv]}\Delta(\uu,\vv;\ww)\le 90n^{3/2}\log n.
\] 
Combining this upper bound with the lower bound of Statement (ii) of Proposition \ref{prop:identify-class} gives the claimed upper bound on $|\Tt_\alpha[\uu,\vv]|$.
\end{proof}

To implement Step 3 of Algorithm \textsc{ComputePairs},
the strategy is to consider each $\alpha$ separately and perform simultaneous quantum searches over $\Tt_\alpha[\uu,\vv]$, as outlined in Figure~\ref{fig:Step3}. We first describe in Section~\ref{subsec:analysis1} how to implement these quantum searches in $\tilde O(n^{1/4})$ rounds for the case $\alpha=0$, and then in Section \ref{subsec:analysis1} how to achieve the same complexity for the case $\alpha>0$.

\begin{figure}[ht!]
\begin{center}
\fbox{
\begin{minipage}{15 cm} 
\begin{itemize}
\item[3.1.]
The nodes apply Algorithm \textsc{IdentifyClass}.
\item[3.2.] For each $\alpha$ do:\vspace{-3mm}

\begin{itemize}
\item[]
Each node $k=(\uu,\vv,x)$ executes, for each $\ell\in [m]$,  a quantum search to check if there exists some $\ww\in \Tt_\alpha[\uu,\vv]$ such that there exists $w\in \ww$ for which $(u^k_\ell,v^k_\ell,w)$ is a negative triangle. The node $k$ outputs all the pairs $\{u^k_\ell,v^k_\ell\}$ for which the quantum search is successful. 
\end{itemize}
\end{itemize}
\end{minipage}
}
\end{center}\vspace{-6mm}
\caption{Details of the implementation of Step 3 of Algorithm \textsc{ComputePairs}.}\label{fig:Step3}
\end{figure}

\subsubsection{Analysis of Step 3.2 for $\boldsymbol{\alpha=0}$}\label{subsec:analysis1}
In Step 3.2 each node $k=(\uu,\vv,x)$ executes~$m$ simultaneous quantum searches. In order to describe this process using the framework presented in Section~\ref{sec:quantum}, with $X=\Tt_0[\uu,\vv]$ and $m=100 n\log n$, we need to explain the evaluation procedure. Since $\Tt_0[\uu,\vv]\subseteq \Vv'$, we have $|\Tt_0[\uu,\vv]|\le \sqrt{n}$. For simplicity (but without loss of generality) we assume below that $|\Tt_0[\uu,\vv]|= \sqrt{n}$.
Observe that the evaluation procedure should implement the following test: each node $k$, when evaluating a list $(\ww^{k}_1,\ldots,\ww^{k}_{m})$ of $m$ elements in $\Tt_0[\uu,\vv]$,  should check for each $\ell\in[m]$ whether there exists a vertex $w\in \ww^k_{\ell}$ such that $\{u^k_\ell,v^k_\ell,w\}$ is a negative triangle. 

Let $L^k_{\ww}\subseteq\Pp(\uu,\vv)$ denote the list consisting of all the pairs $\{u_i^k,v_i^k\}$ such that $\ww^k_i=\ww$, for each node $k=(\uu,\vv,x)$ and each $\ww\in\Tt_0[\uu,\vv]$. We make the assumption that $|L^k_{\ww}|\le 800\sqrt{n}\log n$
for all $k=(\uu,\vv,x)$ and all $\ww\in\Tt_0[\uu,\vv]$ and describe an evaluation procedure that works under this assumption. The procedure is described in Figure \ref{fig:checking2}. 
  
\begin{figure}[ht!]
\begin{center}
\fbox{
\begin{minipage}{15.5 cm} 
\textrm{Input:} each node $k=(\uu,\vv,x)$ receives a list $(\ww^{k}_1,\ldots,\ww^{k}_{m})$ of $m$ elements from $\Tt_0[\uu,\vv]$\\
\textrm{Promise:} the inequality $|L^k_{\ww}|\le 800\sqrt{n}\log n$ holds for each $k=(\uu,\vv,x)$ and each $\ww\in\Tt_0[u,v]$\\
\textrm{Output:} each node $k=(\uu,\vv,x)$ decides, for each $\ell\in[m]$, whether there exists $w\in \ww^k_{\ell}$\\ \vspace{-5mm}

\hspace{14mm} such that $\{u^k_\ell,v^k_\ell,w\}$ is a negative triangle
\begin{itemize}
\item[1.]
Each node $k=(\uu,\vv,x)$ sends the list $L^k_{\ww}$ to node $(\uu,\vv,\ww)$, for each $\ww\in\Tt_0[\uu,\vv]$. Together to each pair $(u,v)$ sent, its weight $f(u,v)$ is also sent.
\item[2.]
Each node $(\uu,\vv,\ww)\in \Tt_0$ checks, for each pair $\{u,v\}$ received at Step 1, whether the inequality
\begin{equation}\label{ineq:min}
\min_{w\in \ww}\{f(u,w)+f(w,v)\} \le f(u,v)
\end{equation}
holds and sends back this information to the node who sent this pair.
\end{itemize}
\end{minipage}
}
\end{center}\vspace{-7mm}
\caption{Evaluation procedure (in the case $\alpha=0$) for the quantum searches implemented at Step~3.2 of \textsc{ComputePairs}.}\label{fig:checking2}
\end{figure}

The procedure of Figure \ref{fig:checking2} obviously always outputs the correct answers since for each $k$ and each $\ell\in[m]$, Inequality (\ref{ineq:min}) at Step 2 precisely checks if there exists  some $w\in \ww_\ell^k$ such that $\{u_\ell^k,v_\ell^k,w\}$ is a negative triangle (because the pair $\{u_\ell^k,v_\ell^k\}$ is sent to node $(\uu,\vv,\ww_\ell^k)$). We now analyze its complexity. Since each list $L^k_{\ww}$ contains at most $800\sqrt{n}\log n$ elements, at Step 1 each node $k=(\uu,\vv,x)$ sends at most $800\sqrt{n}\log n$ elements to $(\uu,\vv,\ww)$ for each $\ww\in \Tt_0[\uu,\vv]$. 
Conversely,  each node $(\uu,\vv,\ww)\in \Tt_0$ receives at most $800\sqrt{n}\log n$ elements from $(\uu,\vv,x)$ for each $x\in[\sqrt{n}]$. Thus, in the CONGEST-CLIQUE model, Step 1 can be implemented in $O(\log n)$ rounds.
Testing whether Inequality (\ref{ineq:min}) holds or not at Step 2 can be done locally using the information collected at Step~1 of Algorithm \textsc{ComputePairs}. Sending back the information at Step 2 can be done with the same complexity as in Step 1. The complexity of the checking procedure is thus $O(\log n)$ rounds.

We can  apply Theorem \ref{th:search} with $X=\Tt_0[\uu,\vv]$ and $\beta=800\sqrt{n}\log n$.  Lemma~\ref{lemma:cond2} guarantees that with probability at least $1-1/n^2$, the assumptions in the statement of Theorem~\ref{th:search} are satisfied. 
Theorem~\ref{th:search} then implies that for $\alpha=0$ the quantum searches of Step 3.2 of Algorithm \textsc{ComputePairs} can be implemented in $\tilde O(n^{1/4})$ rounds and succeed with probability at least $1-2/m^2$. 

 
\subsubsection{Analysis of Step 3.2 for $\boldsymbol{\alpha>0}$}
The analysis of the complexity of the approach presented in Section \ref{subsec:analysis1} crucially relied on the inequality from Lemma \ref{lemma:cond2}, which guarantees that $|\Lambda_x(\uu,\vv)\cap \Delta(\uu,\vv;\ww)|\le 100\cdot\sqrt{n}\log n$. For $\alpha>0$ and $(\uu,\vv,\ww)\in\Tt_\alpha$, Lemma~\ref{lemma:cond2} only gives the weaker upper bound 
\begin{equation}\label{ineq:alpha}
|\Lambda_x(\uu,\vv)\cap \Delta(\uu,\vv;\ww)|\le 100\cdot 2^\alpha\sqrt{n}\log n.
\end{equation} 
The upper bound from Lemma \ref{lemma:ub} is the key observation that will make possible to solve this technical issue.

In Section \ref{subsec:analysis1} each node $(\uu,\vv,x)$ communicated with node $(\uu,\vv,\ww)$ for each $\ww\in\Tt_0[\uu,\vv]$. We used the upper bound $\Tt_0[\uu,\vv]\le \sqrt{n}$ in the analysis. In the case $\alpha>0$ we can use the better upper bound from Lemma \ref{lemma:ub}, which reduces the number of destination nodes by (roughly) a factor~$2^\alpha$. In consequence, we can increase the bandwidth towards these destinations nodes by (roughly) a factor~$2^\alpha$. (This can be done by duplicating the information owned by the destination nodes.) We will show that this is enough to counterbalance the increase by a factor $2^\alpha$ of the message size due to Inequality~(\ref{ineq:alpha}).

We now give more details about the idea of duplicating information to increase the bandwidth. We introduce a new labeling scheme. For the ease of presentation let us assume that $2^\alpha/(720\log n)$ is an integer (if this is not the case the scheme just need to be slightly adapted). In the new scheme each node is assigned a distinct label in $(\uu,\vv,\ww,y)\in \Tt_\alpha\times [2^\alpha/(720\log n)]$. Lemma \ref{lemma:ub} ensures that this can be done. 

Similarly to Section \ref{subsec:analysis1}, let $L^k_{\ww}\subseteq\Pp(\uu,\vv)$ denote the list consisting of all the pairs $\{u_i^k,v_i^k\}$ such that $\ww^k_i=\ww$, for each node $k=(\uu,\vv,x)$ and each $\ww\in\Tt_\alpha[\uu,\vv]$. We make the assumption that $|L^k_{\ww}|\le 800\cdot 2^\alpha\sqrt{n}\log n$
for all $\ww\in\Tt_\alpha[\uu,\vv]$ and describe an evaluation procedure that works under this assumption. The procedure is described in Figure \ref{fig:checking3}. The main difference with the procedure in Section \ref{subsec:analysis1} is that instead of sending the whole list we divide it in sublists and send the sublist $L^k_{\ww,y}$ to $(\uu,\vv,\ww,y)$ for each $y\in[2^\alpha/(720\log n)]$. Another difference is Step 0: each node $(\uu,\vv,\ww)$ first duplicates its input by broadcasting it to all the nodes $(\uu,\vv,\ww,y)$, which can be done in $O(n^{1/4})$ rounds using a randomized routing scheme.

\begin{figure}[ht!]
\begin{center}
\fbox{
\begin{minipage}{15.5 cm} 
\textrm{Input:} each node $k=(\uu,\vv,x)$ receives a list $(\ww^{k}_1,\ldots,\ww^{k}_{m})$ of $m$ elements from $\Tt_\alpha[\uu,\vv]$\\
\textrm{Promise:} the inequality $|L^k_{\ww}|\le 800\cdot 2^\alpha\sqrt{n}\log n$ holds for each $k=(\uu,\vv,x)$ and 

\hspace{15mm} each $\ww\in\Tt_\alpha[u,v]$\\
\textrm{Output:} each node $k=(\uu,\vv,x)$ decides, for each $\ell\in[m]$, whether there exists $w\in \ww^k_{\ell}$\\ \vspace{-5mm}

\hspace{14mm} such that $\{u^k_\ell,v^k_\ell,w\}$ is a negative triangle
\begin{itemize}
\item[0.]
Each node $(\uu,\vv,\ww)\in\Tt_\alpha$ broadcasts the information loaded at Step 1 of Algorithm \textsc{ComputePairs} to the nodes  $(\uu,\vv,\ww,y)$ for all $y\in [2^\alpha/(720\log n)]$.
\item[1.]
Each node $k=(\uu,\vv,x)$ divides, for each $\ww\in\Tt_\alpha[\uu,\vv]$, the list $L^k_{\ww}$ into sublists 
\[
L^k_{\ww,1}, L^k_{\ww,2}, \ldots, L^k_{\ww,2^\alpha/(720\log n)}
\] each containing $O(\sqrt{n}(\log n)^2)$ elements. For each $\ww\in\Tt_\alpha[\uu,\vv]$ and each $y\in[2^\alpha/(720\log n)]$, node $k$ sends the sublist $L^k_{\ww,y}$ to node $(\uu,\vv,\ww,y)$. Together to each pair $(u,v)$ sent, its weight $f(u,v)$ is also sent.

\item[2.]
Each node $(\uu,\vv,\ww,y)\in \Tt_\alpha \times [2^\alpha/(720\log n)]$ checks, for each pair $(u,v)$ received at Step 1, whether the inequality
\[
\min_{w\in \ww}\{f(u,w)+f(w,v)\} \le f(u,v)
\]
holds and sends back this bit of information to the node who sent this pair.
\end{itemize}
\end{minipage}
}
\end{center}\vspace{-6mm}
\caption{Evaluation procedure (in the case $\alpha>0$) for the quantum searches implemented at Step~4 of \textsc{ComputePairs}.}\label{fig:checking3}
\end{figure}

We now analyze the complexity of Steps 1 and 2 of the evaluation procedure. Since each list $L^k_{\ww,y}$ contains at most $O(\sqrt{n}(\log n)^2)$ elements, at Step 1 each node $k=(\uu,\vv,x)$ sends a list containing $O(\sqrt{n}(\log n)^2)$ elements to 
\[
|\Tt_\alpha[\uu,\vv]|\times (2^\alpha/(720\log n))\le \sqrt{n}
\]
nodes (here we used Lemma (\ref{lemma:ub})).
Conversely, each node $(\uu,\vv,\ww,y)\in \Tt_{\alpha}\times [2^\alpha/(720\log n)]$ receives $O(\sqrt{n}(\log n)^2)$ elements from $(\uu,\vv,x)$ for each $x\in[\sqrt{n}]$. Thus, in the CONGEST-CLIQUE model, Step 1 can be implemented in $O((\log n)^2)$ rounds.
Testing whether Inequality (\ref{ineq:min}) holds or not at Step 2 can be done locally using the information collected at Step~1. Sending back the information at the end of Step~2 can be done with the same complexity as in Step 1. The complexity of the checking procedure is thus $O((\log n)^2)$ rounds.

We can apply Theorem \ref{th:search} with $X=\Tt_\alpha[\uu,\vv]$ and $\beta=800\cdot 2^\alpha\sqrt{n}\log n$. Lemma~\ref{lemma:cond2} guarantees that with probability at least $1-1/n^2$, the assumptions in the statement of Theorem~\ref{th:search} are satisfied.  Theorem~\ref{th:search} then implies that for $\alpha>0$ as well the quantum searches of Step 4 of Algorithm \textsc{ComputePairs} can be implemented in $\tilde O(n^{1/4})$ rounds and succeed with probability at least $1-2/m^2$. 


\section*{Acknowledgements}
TI was partially supported by JST SICORP and JSPS KAKENHI grants No. 16H02878 and No. 19K11824.
FLG was partially supported by JSPS KAKENHI grants No.~15H01677, No.~16H01705, No.~16H05853 and No.~19H04066.

\begin{thebibliography}{10}

\bibitem{BKKL17}
Ruben Becker, Andreas Karrenbauer, Sebastian Krinninger, and Christoph Lenzen.
\newblock Near-optimal approximate shortest paths and transshipment in
  distributed and streaming models.
\newblock In {\em Proceedings of the International Symposium on Distributed
  Computing {(DISC)}}, pages 7:1--7:16, 2017.

\bibitem{BN18}
Aaron Bernstein and Danupon Nanongkai.
\newblock Distributed exact weighted all-pairs shortest paths in near-linear
  time.
\newblock In {\em Proceedings of the 51st ACM Symposium on Theory of
  Computing}, 2019 (to appear).
\newblock ArXiv:1811.03337.

\bibitem{Broadbent+08}
Anne Broadbent and Alain Tapp.
\newblock Can quantum mechanics help distributed computing?
\newblock {\em SIGACT News}, 39(3):67--76, 2008.

\bibitem{CKKLPS15}
Keren Censor-Hillel, Petteri Kaski, Janne~H. Korhonen, Christoph Lenzen, Ami
  Paz, and Jukka Suomela.
\newblock Algebraic methods in the congested clique.
\newblock {\em Distributed Computing}, March 2016.

\bibitem{Chang+SODA19}
Yi-Jun Chang, Seth Pettie, and Hengjie Zhang.
\newblock Distributed triangle detection via expander decomposition.
\newblock In {\em Proceedings of the ACM-SIAM Symposium on Discrete Algorithms
  (SODA)}, pages 821--840, 2019.

\bibitem{Chang+19}
Yi-Jun Chang and Thatchaphol Saranurak.
\newblock Improved distributed expander decomposition and nearly optimal
  triangle enumeration.
\newblock ArXiv:1904.08037, April 2019.

\bibitem{Denchev+08}
Vasil~S. Denchev and Gopal Pandurangan.
\newblock Distributed quantum computing: a new frontier in distributed systems
  or science fiction?
\newblock {\em SIGACT News}, 39(3):77--95, 2008.

\bibitem{Dolev+DISC12}
Danny Dolev, Christoph Lenzen, and Shir Peled.
\newblock ``{Tri, Tri Again}": Finding triangles and small subgraphs in a
  distributed setting - (extended abstract).
\newblock In {\em Proceedings of the International Symposium on Distributed
  Computing (DISC)}, pages 195--209, 2012.

\bibitem{Drucker+PODC14}
Andrew Drucker, Fabian Kuhn, and Rotem Oshman.
\newblock On the power of the congested clique model.
\newblock In {\em Proceedings of the ACM Symposium on Principles of Distributed
  Computing (PODC)}, pages 367--376, 2014.

\bibitem{Elkin17}
Michael Elkin.
\newblock Distributed exact shortest paths in sublinear time.
\newblock In {\em Proceedings of the ACM Symposium on Theory of Computing
  (STOC)}, pages 757--770, 2017.

\bibitem{Elkin+PODC14}
Michael Elkin, Hartmut Klauck, Danupon Nanongkai, and Gopal Pandurangan.
\newblock Can quantum communication speed up distributed computation?
\newblock In {\em Proceedings of the ACM Symposium on Principles of Distributed
  Computing (PODC)}, pages 166--175, 2014.

\bibitem{FN18}
Sebastian Forster and Danupon Nanongkai.
\newblock A faster distributed single-source shortest paths algorithm.
\newblock In {\em Proceedings of the IEEE Symposium on Foundations of Computer
  Science (FOCS)}, pages 686--697, 2018.

\bibitem{FHW12}
Silvio Frischknecht, Stephan Holzer, and Roger Wattenhofer.
\newblock Networks cannot compute their diameter in sublinear time.
\newblock In {\em Proceedings of the ACM-SIAM Symposium on Discrete Algorithms
  (SODA)}, pages 1150--1162, 2012.

\bibitem{Gavoille+DISC09}
Cyril Gavoille, Adrian Kosowski, and Marcin Markiewicz.
\newblock What can be observed locally?
\newblock In {\em Proceedings of the International Symposium on Distributed
  Computing (DISC)}, pages 243--257, 2009.

\bibitem{GL18}
Mohsen Ghaffari and Jason Li.
\newblock Improved distributed algorithms for exact shortest paths.
\newblock In {\em Proceedings of the ACM Symposium on Theory of Computing
  (STOC)}, pages 431--444, 2018.

\bibitem{Ghaffari+18}
Mohsen Ghaffari and Krzysztof Nowicki.
\newblock Congested clique algorithms for the minimum cut problem.
\newblock In {\em Proceedings of the ACM Symposium on Principles of Distributed
  Computing (PODC)}, pages 357--366, 2018.

\bibitem{GroverSTOC96}
Lov~K. Grover.
\newblock A fast quantum mechanical algorithm for database search.
\newblock In {\em Proceedings of the ACM Symposium on Theory of Computing
  (STOC)}, pages 212--219, 1996.

\bibitem{Hegeman+PODC15}
James~W. Hegeman, Gopal Pandurangan, Sriram~V. Pemmaraju, Vivek~B. Sardeshmukh,
  and Michele Scquizzato.
\newblock Toward optimal bounds in the congested clique: Graph connectivity and
  {MST}.
\newblock In {\em Proceedings of the ACM Symposium on Principles of Distributed
  Computing (PODC)}, pages 91--100, 2015.

\bibitem{Hegeman+SIROCCO14}
James~W. Hegeman and Sriram~V. Pemmaraju.
\newblock Lessons from the congested clique applied to {MapReduce}.
\newblock In {\em Proceedings of the International Colloqium on Structural
  Information and Communication Complexity (SIROCCO)}, pages 149--164, 2014.

\bibitem{Hegeman+DISC14}
James~W. Hegeman, Sriram~V. Pemmaraju, and Vivek Sardeshmukh.
\newblock Near-constant-time distributed algorithms on a congested clique.
\newblock In {\em Proceedings of the International Symposium on Distributed
  Computing (DISC)}, pages 514--530, 2014.

\bibitem{HKN16}
Monika Henzinger, Sebastian Krinninger, and Danupon Nanongkai.
\newblock A deterministic almost-tight distributed algorithm for approximating
  single-source shortest paths.
\newblock In {\em Proceedings of the ACM Symposium on Theory of Computing
  (STOC)}, pages 489--498, 2016.

\bibitem{HW12}
Stephan Holzer and Roger Wattenhofer.
\newblock Optimal distributed all pairs shortest paths and applications.
\newblock In {\em Proceedings of the ACM Symposium on Principles of Distributed
  Computing (PODC)}, pages 355--364, 2012.

\bibitem{HNS17}
Chien{-}Chung Huang, Danupon Nanongkai, and Thatchaphol Saranurak.
\newblock Distributed exact weighted all-pairs shortest paths in
  $\tilde{O}(n^{5/4})$ rounds.
\newblock In {\em Proceedings of the IEEE Symposium on Foundations of Computer
  Science (FOCS)}, pages 168--179, 2017.

\bibitem{Izumi+PODC17}
Taisuke Izumi and Fran\c{c}ois Le~Gall.
\newblock Triangle finding and listing in {CONGEST} networks.
\newblock In {\em Proceedings of the ACM Symposium on Principles of Distributed
  Computing (PODC)}, pages 381--389, 2017.

\bibitem{Jurdzinski+18}
Tomasz Jurdzi\'nski and Krzysztof Nowicki.
\newblock {MST} in {$O(1)$} rounds of congested clique.
\newblock In {\em Proceedings of the ACM-SIAM Symposium on Discrete Algorithms
  (SODA)}, pages 2620--2632, 2018.

\bibitem{LeGall16}
Fran{\c{c}}ois {Le Gall}.
\newblock Further algebraic algorithms in the congested clique model and
  applications to graph-theoretic problems.
\newblock In {\em Proceedings of the International Symposium on Distributed
  Computing (DISC)}, pages 57--70, 2016.

\bibitem{LeGall+PODC18}
Fran{\c{c}}ois {Le Gall} and Fr{\'e}d{\'e}ric Magniez.
\newblock Sublinear-time quantum computation of the diameter in {CONGEST}
  networks.
\newblock In {\em Proceedings of the ACM Symposium on Principles of Distributed
  Computing (PODC)}, pages 337--346, 2018.

\bibitem{LeGall+STACS19}
Fran{\c{c}}ois {Le Gall}, Harumichi Nishimura, and Ansis Rosmanis.
\newblock Quantum advantage for the {LOCAL} model in distributed computing.
\newblock In {\em Proceedings of the International Symposium on Theoretical
  Aspects of Computer Science (STACS)}, pages 49:1--49:14, 2019.

\bibitem{LP13}
Christoph Lenzen and David Peleg.
\newblock Efficient distributed source detection with limited bandwidth.
\newblock In {\em Proceedings of the ACM Symposium on Principles of Distributed
  Computing (PODC)}, pages 375--382, 2013.

\bibitem{Lenzen+STOC11}
Christoph Lenzen and Roger Wattenhofer.
\newblock Tight bounds for parallel randomized load balancing: extended
  abstract.
\newblock In {\em Proceedings of the ACM Symposium on Theory of Computing
  (STOC)}, pages 11--20, 2011.

\bibitem{Lotker+SPAA03}
Zvi Lotker, Elan Pavlov, Boaz Patt{-}Shamir, and David Peleg.
\newblock {MST} construction in $o(\log \log n)$ communication rounds.
\newblock In {\em Proceedings of the Symposium on Parallel Algorithms and
  Architectures (SPAA)}, pages 94--100, 2003.

\bibitem{Nanongkai14}
Danupon Nanongkai.
\newblock Distributed approximation algorithms for weighted shortest paths.
\newblock In {\em Proceedings of the ACM Symposium on Theory of Computing
  (STOC)}, pages 565--573, 2014.

\bibitem{Pandurangan+SPAA18}
Gopal Pandurangan, Peter Robinson, and Michele Scquizzato.
\newblock On the distributed complexity of large-scale graph computations.
\newblock In {\em Proceedings of the Symposium on Parallel Algorithms and
  Architectures (SPAA)}, pages 405--414, 2018.

\bibitem{parter18}
Merav Parter.
\newblock {$(\Delta+1)$}-coloring in the congested clique model.
\newblock In {\em Proceedings of the International Colloquium on Automata,
  Languages and Programming (ICALP)}, pages 160:1--160:14, 2018.

\bibitem{PattShamir+PODC11}
Boaz Patt{-}Shamir and Marat Teplitsky.
\newblock The round complexity of distributed sorting: extended abstract.
\newblock In {\em Proceedings of the ACM Symposium on Principles of Distributed
  Computing (PODC)}, pages 249--256, 2011.

\bibitem{Peleg00}
David Peleg.
\newblock {\em Distributed computing: a locality-sensitive approach}.
\newblock Society for Industrial and Applied Mathematics, 2000.

\bibitem{SHKKANPPW12}
Atish~Das Sarma, Stephan Holzer, Liah Kor, Amos Korman, Danupon Nanongkai,
  Gopal Pandurangan, David Peleg, and Roger Wattenhofer.
\newblock Distributed verification and hardness of distributed approximation.
\newblock {\em {SIAM} Journal on Computing}, 41(5):1235--1265, 2012.

\bibitem{Tani+12}
Seiichiro Tani, Hirotada Kobayashi, and Keiji Matsumoto.
\newblock Exact quantum algorithms for the leader election problem.
\newblock {\em {ACM} Transactions on Computation Theory}, 4(1):1:1--1:24, 2012.

\bibitem{Williams+JACM18}
Virginia~Vassilevska Williams and Ryan Williams.
\newblock Subcubic equivalences between path, matrix, and triangle problems.
\newblock {\em Journal of the {ACM}}, 65(5):27:1--27:38, 2018.

\bibitem{deWolf02}
{Ronald de} Wolf.
\newblock Quantum communication and complexity.
\newblock {\em Theoretical Computer Science}, 287(1):337--353, 2002.

\bibitem{ZwickJACM02}
Uri Zwick.
\newblock All pairs shortest paths using bridging sets and rectangular matrix
  multiplication.
\newblock {\em Journal of the {ACM}}, 49(3):289--317, 2002.

\end{thebibliography}

\appendix

\section{Distributed multiple quantum searches}\label{appendix:quantum}
In this appendix we prove Theorem \ref{th:search}. 

Let $\Qq$ denote the $\tilde O(r\sqrt{|X|})$-round quantum algorithm described at the end of Section \ref{sub:LGM}. Remember that this algorithm implements in parallel $m$ independent executions of Grover's algorithm and uses Algorithm~$\Cc_m$ as a global evaluation procedure.
We first analyze this algorithm in more details.
For any string $b\in\{0,1\}^m$ let us define the quantum state
\[
\ket{\psi^b}=\ket{\psi^{b_1}_1}\otimes\cdots\otimes \ket{\psi^{b_m}_m},
\]
where 
\[
\ket{\psi_i^0}=\frac{1}{\sqrt{|A_i^0|}}\sum_{x\in A_i^0}\ket{x}\hspace{5mm}\textrm{ and }\hspace{5mm}
\ket{\psi_i^1}=\frac{1}{\sqrt{|A_i^1|}}\sum_{x\in A_i^1}\ket{x}
\]
for each $i\in[m]$. Let $\Hh_m$ denote the Hilbert space spanned by all the quantum states in the set $\{\ket{\psi^b}\}_{b\in\{0,1\}^n}$. An important observation is that Algorithm $\Qq$ leaves the space $\Hh_m$ invariant. The initial state of Algorithm $\Qq$ is
\[
\ket{\Phi^m_0}=\frac{1}{\sqrt{|X|^m}}\sum_{(x_1,\ldots,x_m)\in X^m}\ket{x_1}\otimes\cdots\otimes\ket{x_m},
\]
which is in $\Hh_m$. For each $k\ge 0$, one step of the algorithm maps the state $\ket{\Phi^m_{k}}$ to the state
\[
\ket{\Phi^m_{k+1}} = U_m\Cc_m \ket{\Phi^m_{k}},
\]
where $U_m$ is a unitary operator independent of the function $g_1,\ldots,g_n$ and $\Cc_m$ represents the unitary operator corresponding to the quantum circuit obtained by converting the classical algorithm $\Cc_m$ into a quantum circuit. Analyzing Grover's algorithm shows that after $k=O(\sqrt{|X|})$ iterations the quantum state $\ket{\Phi_{k}^m}$ becomes close to the state $\ket{\psi_1^1}\otimes\cdots\otimes\ket{\psi_m^1}$, and thus measuring this state gives an element from $A_1^1\times\cdots\times A_m^1$ with high probability. This success probability can be amplified to (for instance) $1-1/m^2$ by repeating the algorithm a logarithmic number of time.

Let $\tilde \Qq$ be exactly the same algorithm as $\Qq$ but with each application of the quantum circuit corresponding to $\Cc_m$ replaced by an application of the quantum circuit corresponding to $\tilde \Cc_m$. Let us analyze the output of $\tilde \Qq$. As in Section \ref{sub:LGM}, we make the assumption $|A^1_i|\le |X|/2$, for all $i\in[m]$.
Let $\Hh'_m$ denote the Hilbert space spanned by all vectors $\ket{x_1}\otimes\cdots\otimes\ket{x_m}$ with $(x_1,\ldots,x_m)\in\Upsilon_\beta(m,X)$, and $\Hh''_m$ denote the Hilbert space spanned by all $\ket{x_1}\otimes\cdots\otimes\ket{x_m}$ with $(x_1,\ldots,x_m)\in X^m\setminus\Upsilon_\beta(m,X)$. Let $\Pi_m$ denote the projection into $\Hh''_m$. We first show the following crucial lemma.
\begin{lemma}\label{lemma:state}
Assume that 
$\beta>8m/|X|$ and
$A_1^1\times\cdots\times A_m^1\subseteq \Upsilon_{\beta/2}(m,X)$.
For any quantum state $\ket{\varphi}\in\Hh_m$ we have 
\[
\Big\|
\Pi_m \ket{\varphi}
\Big\|^2
<|X|\times \exp\left(-\frac{2m}{9|X|}\right).
\]
\end{lemma}
\begin{proof}
The state $\ket{\varphi}$ can be written as
\[
\ket{\varphi}=\sum_{b\in\{0,1\}^m}\alpha_b \ket{\psi^b}
\]
for some amplitude $\alpha_b\in\mathbb{C}$ such that $\sum_{b\in\{0,1\}^m}|\alpha_b|^2=1$. Observe that
\[
\Big\|
\Pi_m \ket{\varphi}
\Big\|^2
=
\sum_{b\in\{0,1\}^m}|\alpha_b|^2 \Big\|
\Pi_m \ket{\psi^b}
\Big\|^2
\]
since all the vectors $\Pi_m \ket{\psi^b}$ are orthogonal.
We show below that the inequality
\begin{equation}\label{ineq1}
\Big\|
\Pi_m \ket{\psi^b}
\Big\|^2
<|X|\times \exp\left(-\frac{2m}{9|X|}\right)
\end{equation}
holds for any $b\in\{0,1\}^m$. The claimed upper bound on $\|\Pi_m\ket{\varphi}\|$ then immediately follows. 

Consider a string $b\in\{0,1\}^m$.
Let us assume, without loss of generality, that $b$ is the string with~$0$s in the first $\ell$ positions, followed by $1$s in the next $m-\ell$ positions, for some integer $\ell\in\{0,1,\ldots,m\}$. The state $\ket{\psi^b}$ is thus the uniform superposition of all the states $\ket{x_1}\otimes\cdots\otimes\ket{x_m}$ for all $(x_1,\ldots,x_\ell,x_{\ell+1},\ldots,x_m)\in A^0_1\times \cdots\times A^0_\ell\times A^1_{\ell+1}\times \cdots\times A^1_m$. For any choice of $(x_{\ell+1},\ldots,x_m)\in A^1_{\ell+1}\times \cdots\times A^1_m$, we claim that the fraction of $(x_1,\ldots,x_{\ell})\in A^0_{1}\times \cdots \times A^0_\ell$ such that $(x_1,\ldots,x_m)\notin \Upsilon_\beta(m,M)$ is at most 
\[
|X|\times \exp\left(-\frac{2m}{9|X|}\right).
\] 
This immediately implies Inequality (\ref{ineq1}).

Let us prove the claim. Remember that we are assuming $A_1^1\times\cdots\times A_m^1\subseteq \Upsilon_{\beta/2}(m,X)$.
For any $x\in X$, we thus know that there are at most $\beta/2$ indices $i\in\{\ell+1,\ldots,m\}$ such that $x_i=x$. For each $i\in \{1,\ldots, \ell\}$, the probability that an element taken uniformly at random from $A_i^0$ equals $x$ is at most $1/|A_i^0|\le 2/|X|$. When $(x_{1},\ldots,x_{\ell})$ is chosen uniformly at random in $A_{1}^0\times\cdots\times A_\ell^0$, the expected number of times $x$ appears is thus at most 
\[
\frac{2\ell}{|X|}+\frac{\beta}{2}\le \frac{2m}{|X|}+\frac{\beta}{2}<\frac{3}{4}\beta,
\]
where we used the assumption $\beta>8m/|X|$ for the last inequality.
Chernoff's bound implies that the probability that $x$ appears more than $\beta$ times is at most 
\[
\exp\left(-\frac{2m}{9|X|}\right),
\]
and the claim then follows from the union bound. 
\end{proof}

We can now analyze the output of Algorithm $\tilde \Qq$ and prove Theorem \ref{th:search}.
\begin{proof}[Proof of Theorem \ref{th:search}]
Let $\ket{\tilde{\Phi}_k^m}$ denote the state at the $k$-th iteration when executing Algorithm $\tilde \Qq$. Initially we have $\ket{\tilde{\Phi}_0^m}=\ket{{\Phi}_0^m}$. For any $k\ge 0$ let us write
\[
\ket{{\Phi}_k^m}=\ket{{\Phi}'_k}+\ket{{\Phi}''_k}
\textrm{ and }
\ket{\tilde{\Phi}_k^m}=\ket{\tilde{\Phi}'_k}+\ket{\tilde{\Phi}''_k},
\]
where $\ket{{\Phi}'_k}$ and $\ket{{\Phi}''_k}$ are the projections of $\ket{{\Phi}_k^m}$ into $\Hh_m'$ and $\Hh_m''$, respectively, and $\ket{{\tilde\Phi}'_k}$ and $\ket{{\tilde \Phi}''_k}$ are the projections of $\ket{{\tilde\Phi}_k^m}$ into $\Hh_m'$ and $\Hh_m''$, respectively. Note that $\Cc_m\ket{\Phi_{k}'}=\tilde\Cc_m\ket{\Phi_{k}'}$ for all $k\ge 0$.

We have
\begin{align*}
\Big\|
\ket{\Phi_{k+1}^m}
- 
\ket{\tilde{\Phi}_{k+1}^m}
\Big\|
&=
\Big\|
U_m\Cc_m\ket{\Phi_{k}^m}
- 
U_m\tilde \Cc_m\ket{\tilde{\Phi}_{k}^m}
\Big\|\\
&=
\Big\|
\Cc_m\ket{\Phi_{k}^m}
- 
\tilde\Cc_m\ket{\tilde{\Phi}_{k}^m}
\Big\|\\
&=
\Big\|
(\Cc_m\ket{\Phi^m_{k}}-\tilde \Cc_m\ket{ {\Phi}^m_{k}})
+
\tilde\Cc_m(\ket{\Phi^m_{k}}-\ket{\tilde\Phi^m_{k}})
\Big\|\\
&\le
\Big\|
\Cc_m\ket{\Phi^m_{k}}-\tilde \Cc_m\ket{ {\Phi}^m_{k}}
\Big\|
+
\Big\|
\ket{\Phi^m_{k}}-\ket{ {\tilde\Phi}^m_{k}}
\Big\|\\
&\le \Big\|
\Cc_m\ket{\Phi'_{k}}-\tilde \Cc_m\ket{ {\Phi}'_{k}}
\Big\|
+
\Big\|
\Cc_m\ket{\Phi''_{k}}-\tilde \Cc_m\ket{ {\Phi}''_{k}}
\Big\|
+
\Big\|
\ket{\Phi^m_{k}}-\ket{ {\tilde\Phi}^m_{k}}
\Big\|\\
&\le
2\Big\|
\ket{\Phi''_{k}}
\Big\|
+
\Big\|
\ket{\Phi^m_{k}}-\ket{ {\tilde\Phi}^m_{k}}
\Big\|\\
&\le
2\sqrt{|X|}\times \exp\left(-\frac{m}{9|X|}\right)
+
\Big\|
\ket{\Phi^m_{k}}-\ket{ {\tilde\Phi}^m_{k}}
\Big\|,
\end{align*}
where we used Lemma \ref{lemma:state} to obtain the last inequality.
We conclude that for any $k\ge 0$ we have

\begin{align*}
\Big\|
\ket{\Phi_{k}^m}
- 
\ket{\tilde{\Phi}_{k}^m}
\Big\|
&\le
\Big\|
\ket{\Phi'_{0}}-\ket{\tilde {\Phi}'_{0}}
\Big\|
+
2k\sqrt{|X|}\times \exp\left(-\frac{m}{9|X|}\right)\\
&=
2k\sqrt{|X|}\times \exp\left(-\frac{m}{9|X|}\right)\\
&\le
\frac{2k}{m^3},
\end{align*}
where we used the assumption $|X|<m/(36\log m)$ to derive the last inequality.
This implies that the output of Algorithm $\tilde \Qq$ is the same as the output of Algorithm $\Qq$ with probability at least $1-1/m^2$. The output of $\tilde \Qq$ is thus correct with probability at least $1-2/m^2$, from the union bound.
\end{proof}

\end{document}